\documentclass[11pt]{article}

\usepackage{microtype}

\usepackage{amsthm,amsfonts,graphicx}
\newcommand{\real}{\mathbb{R}}
\let\bd\partial

\usepackage{fullpage}

\usepackage{enumerate}

\usepackage[medium,compact]{titlesec}

\newtheorem{theorem}{Theorem} 
\newtheorem{lemma}[theorem]{Lemma} 
\newtheorem{fact}[theorem]{Fact} 

\theoremstyle{definition}

\newtheorem*{udefinition}{Definition}

\theoremstyle{remark}

\usepackage{url}
\newcommand{\Boris}[2][says]{\textsc{Boris #1}:
  \emph{#2}\ifinner\else\marginpar[\hfill BA
  $\rightarrow$]{$\leftarrow$ BA}\fi}
\newcommand{\Dmitriy}[2][says]{\textsc{Dmitriy #1}:
  \emph{#2}\ifinner\else\marginpar[\hfill DD
  $\rightarrow$]{$\leftarrow$ DD}\fi}

\begin{document}


\author{Boris Aronov\thanks{%
    Department of Computer Science and Engineering,
    Polytechnic Institute of NYU,
    Brooklyn, New York, USA;
    \texttt{aronov@poly.edu}.
    Work of B.A.\ on this paper had been partially supported by a
    grant from the U.S.-Israel Binational Science Foundation and by
    NSA MSP Grant H98230-06-1-0016. B.A.\ blames D.~Halperin for
    provoking him into working on this problem.}
  \and
  Dmitriy Drusvyatskiy\thanks{%
    School of Operations Research and Information Engineering,
    Cornell University,
    Ithaca, New York, USA;
    \texttt{http://people.orie.cornell.edu/dd379}.
    Work of Dmitriy Drusvyatskiy on this paper had been partially supported by the NDSEG grant from the Department of Defense and by the NSF Computatonal Sustainability Grant 0832782.
}}

\title{\Large Complexity of a Single Face in an Arrangement of
  $s$-Intersecting Curves%
  \thanks{This is the version of the paper from June 1, 2009, with
    several typos corrected.}}

\maketitle

\begin{abstract}
  Consider a face $F$ in an arrangement of $n$ Jordan curves in the
  plane, no two of which intersect more than $s$ times.  We prove that
  the combinatorial complexity of $F$ is $O(\lambda_s(n))$,
  $O(\lambda_{s+1}(n))$, and $O(\lambda_{s+2}(n))$, when the curves are
  bi-infinite, semi-infinite, or bounded, respectively;
  $\lambda_k(n)$ is the maximum length of a Davenport-Schinzel
  sequence of order $k$ on an alphabet of $n$ symbols.

  Our bounds asymptotically match the known worst-case lower bounds.
  Our proof settles the still apparently open case of semi-infinite
  curves. Moreover, it treats the three cases in a fairly uniform fashion.
\end{abstract}

 \section{Introduction}
In this paper we study the maximum complexity of a single face in an
arrangement of curves in the plane, no two of which intersect more
than $s$ times; see below. We will do this through an extensive use
of Davenport-Schinzel sequences, which were first introduced by
Davenport and Schinzel in 1965 \cite{original}. They were motivated,
curiously enough, by a problem in differential equations.

\begin{udefinition}
Let $n$, $s$ be positive integers. A sequence $U=\langle
u_1,\ldots,u_m\rangle$ over an alphabet of size $n$ is a
\emph{Davenport-Schinzel sequence} of order $s$ on an alphabet of
$n$ symbols, or \emph{DS(n,s)-sequence}, for short, if it satisfies
the following conditions:
\begin{enumerate}
\item $u_i\neq u_{i+1}$, for each $1\leq i< m$.
\item There do not exist $s+2$ indices
$1\leq i_1<i_2<\ldots<i_{s+2}\leq m$ such that
$u_{i_1}=u_{i_3}=u_{i_5}=\ldots=a$,
$u_{i_2}=u_{i_4}=u_{i_6}=\ldots=b$, for some distinct symbols $a$
and $b$.
\end{enumerate}
We denote by $\lambda_s(n)$ the length of the longest
$DS(n,s)$-sequence.
\end{udefinition}

Davenport and Schinzel were able to establish a connection between
these sequences and lower envelopes of collections of functions
\cite{original,DS-book}. The next significant step in studying these
sequences was taken by Szemer\'edi in 1974, who established improved
upper bounds on the length of Davenport-Schinzel sequences
\cite{Sze}. In 1983, Atalah's work was the first step in
establishing DS-sequences as a fundamental tool in computational and
combinatorial geometry \cite{ata}. The fundamental question that was
still unanswered was determining the asymptotic growth rate of the
functions $\lambda_s(n)$, for any fixed $s$.  For $s=1$ and $s=2$,
this is very easy 
($\lambda_1(n)=n$ and $\lambda_2(n)=2n-1$) but already for $s=3$,
this question is highly nontrivial. In 1986, Hart and Sharir showed
that that the maximum length of any $DS(n,3)$-sequence is
$O(n\alpha(n))$, where $\alpha(n)$ is the very slowly growing
inverse of Ackermann's function \cite{nonlin_DS}. In 1989, Agarwal,
Sharir, and Shor completed this classification by showing nearly
tight, nearly linear bounds for all fixed $s$ \cite{general}.
Davenport-Schinzel sequences have proven to be very useful in
providing tighter methods of analysis for many problems in discrete
and computational geometry \cite{DS-book,survey}.

In this paper, we are interested in the following three types of
curves. An \emph{unbounded Jordan curve} is the image of an open
unit interval under a topological embedding into $\real^2$,
such that it separates the plane. A \emph{semi-infinite Jordan
curve} is the image of a half-open unit interval under a topological
embedding into $\real^2$, such that the image is unbounded with
respect to the standard Euclidean norm. A \emph{bounded Jordan
curve} (or \emph{Jordan arc}) is the image of a closed unit interval
under a topological embedding into $\real^2$.

Let $\Gamma_0$, $\Gamma_1$, and $\Gamma_2$ be collections of $n$
bi-infinite, semi-infinite, and bounded Jordan curves in the plane,
respectively, such that any two curves in $\Gamma_i$ intersect at most
$s$ times, for some fixed constant $s>0$. (The subscript of $\Gamma_i$
signifies the number of finite endpoints of each curve in this
collection.)

\begin{udefinition}\cite{Edels,DS-book}
  The \emph{arrangement} $A(\Gamma_i)$ of $\Gamma_i$ is the planar
  subdivision induced by the arcs of $\Gamma_i$. Thus $A(\Gamma_i)$
  is a planar map whose \emph{vertices} are the endpoints of curves of $\Gamma_i$, if
  any, and their pairwise intersection points. The \emph{edges} are
  maximal connected portions of the curves that do not contain a
  vertex. The \emph{faces} are the connected components of
  $\real^2 - \bigcup \Gamma_i$.
\end{udefinition}

The \emph{combinatorial complexity} of a face $F$ of $A(\Gamma_i)$
is the total number of vertices and edges of $A(\Gamma_i)$ along its
boundary $\bd F$. A feature on $\bd F$ is counted in the
complexity as many times as it appears.

We are interested in studying the maximum combinatorial complexity of a
single face of $A(\Gamma_i)$.  Schwartz and Sharir showed that the
combinatorial complexity of a single face of $A(\Gamma_0)$ is at
most $\lambda_{s}(n)$ \cite{unbounded}.  Sharir \emph{et al.}\@ showed that
the combinatorial complexity of a single face of $A(\Gamma_2)$ is
$O(\lambda_{s+2}(n))$ \cite{DS-sharir}.  (See Nivasch~\cite{Nivasch}
for some very recent progress in this subject.)
However, the two proofs provided are very different, which is somewhat
unsatisfying. 

It has also been conjectured that the combinatorial complexity of a
single face of $A(\Gamma_1)$ is $O(\lambda_{s+1}(n))$.  There has been
some work to suggest that this is true. For instance, Alevizos,
Boissonnat, and Preparata showed that the complexity of a single face
in an arrangement of rays is linear \cite{rays}; this is the case when
$s=1$. In this paper, we prove the following theorem:

\begin{theorem}
  \label{theorem:main}
  The combinatorial complexity of a face $F$ in an arrangement of $n$
  bi-infinite, semi-infinite, or bounded Jordan curves, no two of which intersect more than $s$ times, is
  $O(\lambda_s(n))$, $O(\lambda_{s+1}(n))$, and $O(\lambda_{s+2}(n))$,
  respectively.
\end{theorem}

These upper bounds are tight in the worst case. This easily follows
from the fact that the complexity of the lower envelope of the
collection of functions defined by curves of $\Gamma_i$ (these
functions are partially defined for $i = 1$ and for $i=2$), in the
special case where the curves are $x$-monotone, is a lower bound on
the combinatorial complexity of a single face of $A(\Gamma_i)$
\cite{ata,DS-book}; the maximum complexity of such an envelope is
$\Theta(\lambda_s(n))$, $\Theta(\lambda_{s+1}(n))$, and
$\Theta(\lambda_{s+2}(n))$, respectively.  

The result is known for bi-infinite curves and for Jordan arcs \cite{DS-sharir, unbounded}, but not, to the best of our
knowledge, for semi-infinite curves.  The advantage of our proof is
that firstly it settles the previously mentioned conjecture and
secondly it treats all three cases in a reasonably uniform manner.

Our paper is organized as follows. In
Section~\ref{section:combinatorial} we prove a purely combinatorial
auxiliary fact (Fact~\ref{fact:combination}).
Section~\ref{section:preliminaries} contains some preliminary
modifications to the geometric problem.  In
Section~\ref{section:unbounded} we prove bounds on the maximum
complexity of an unbounded face in $A(\Gamma_i)$. Finally, in
Section~\ref{section:arbitrary} we transform any bounded face of
$A(\Gamma_i)$ into an unbounded one without an asymptotic increase in
its complexity. This implies bounds on the maximum complexity of a
bounded face in $A(\Gamma_i)$ and yields our main theorem.

\section{A Combinatorial Fact}
\label{section:combinatorial}

In this section, we state and prove a simple combinatorial fact about
Davenport-Schinzel sequences.  It or a close relative have been ``in
the folklore'' of this area of research~\cite{DS-book}, although we
have not been able to pin down a source where it was explicitly stated
in this form.  For completeness, we present a proof.

\begin{udefinition}
  Given a sequence $S$ over an alphabet $\Sigma$, for $\Lambda
  \subseteq \Sigma$, $S_{|\Lambda}$ denotes the sequence obtained by
  deleting from $S$ all symbols not in $\Lambda$.
\end{udefinition}

\begin{udefinition}
  Let $\Sigma$ be an alphabet. Denote
  by $\Sigma^*$ the set of all finite sequences over $\Sigma$. We  define an operation $\diamond: \Sigma^*
  \longrightarrow \Sigma^*$ as follows. Let $X
  \in \Sigma^*$. $X^\diamond$ is obtained from
  $X$ by simply collapsing each subsequence of consecutive identical
  elements to a single element, e.g., $\langle \ldots
  ,a,b,b,b,b,c,c,d,\ldots \rangle$ would be collapsed to $\langle \ldots
  ,a,b,c,d, \ldots \rangle$.
\end{udefinition}

\begin{udefinition}
  Let $\Sigma_1$ and $\Sigma_2$ be disjoint alphabets, let $k\geq 1$ be an integer, and let
  $X=\langle x_1,\ldots,x_m\rangle$ be a sequence over $\Sigma_1 \cup \Sigma_2$. We say that $X$ is
  \emph{k-friendly under} ($\Sigma_1$,$\Sigma_2$) if the following
  condition holds:

  \begin{itemize}
  \item[($\ast$)] There do not exist $k+1$ consecutive indices $1 \leq
    i,i+1, \ldots,i+k \leq m$ such that
    $x_{i}=x_{i+2}=x_{i+4}=\ldots=a$,
    $x_{i+1}=x_{i+3}=x_{i+5}=\ldots=b$, with $a \in \Sigma_1$ and $b
    \in \Sigma_2$, or vice versa.
  \end{itemize}
\end{udefinition}


\begin{fact} \label{fact:combination}
  If a sequence $X$ is k-friendly under $(\Sigma_1$,$\Sigma_2)$, no two consecutive symbols of $X$ are the same, and
  $(X_{|\Sigma_1})^\diamond$ and $(X_{|\Sigma_2})^\diamond$ are both $DS(n,s)$-sequences, then $|X| = O(k\lambda_s
  (n))$.
\end{fact}
\begin{proof}
  Let $L' = X_{|\Sigma_1}$, $L = (L')^\diamond$, $R' = X_{|\Sigma_2}$,
  and $R = (R')^\diamond$.  It is clear that $|X| = |L| + |R| + (|L'|
  - |L|) + (|R'| - |R|)$. Since $|L|, |R| \leq \lambda_s(n)$, without
  loss of generality, it is sufficient to bound $\Delta_L = |L'| -
  |L|$. $\Delta_L$ is the number of elements that were deleted from
  $L'$ by the $\diamond$ operation. Suppose that a subsequence
  $\langle a,b,\ldots,b,c \rangle$ of $|L'|$ was collapsed to $\langle
  a,b,c \rangle$ in $L$ (collapses at the beginning and at the end of $L$ are
  handled similarly). Now the only way that this could have happened
  is that in $X$, between every two corresponding consecutive elements
  $b$, there was a sequence of one or more elements all from $R'$;
  denote such a sequence by $\xi_i$. Let $T = \langle
  b,\xi_1,b,\xi_2,\ldots,b \rangle$. We charge each element $b$ in $T$
  to an element of $\xi_i$ following it, such that if possible it is
  different from the element of $\xi_{i-1}$ that was charged for the
  previous occurrence of $b$. If it were always possible to do so,
  then all the elements of $R'$ that have been charged would be
  preserved when $R'$ were transformed into $R$ and each one would
  have been charged only once, so we could bound $\Delta_L$ by
  $|R|$. The only time that it is not possible is when there is a
  subsequence of $X$ of the form $\langle \ldots,b,r,b,r,\ldots
  \rangle$, where r is an element of $R$. Since $X$ is
  \emph{k-friendly under} $(\Sigma_1,\Sigma_2)$, the length of such a
  subsequence is no larger than $k$. It now easily follows that in the
  above charging scheme, an element of $R$ may be charged up to $O(k)$
  times, so $\Delta_L = O(k\lambda_s(n))$. Therefore $|X| = |L| + |R|
  + \Delta_L + \Delta_R = O(k\lambda_s(n))$.
\end{proof}

\section{Geometric Preliminaries}
\label{section:preliminaries} We now return to the geometric
problem. Recall that we start with a set $\Gamma_i$ of curves in the
plane, no two intersecting pairwise more than $s$ times. In order to
state our argument, no modifications will be required for curves in
$\Gamma_0$, since only one side of any curve in $\Gamma_0$ can appear
on the boundary of $F$. However some modifications will be needed for
the curves in $\Gamma_1$ and $\Gamma_2$, which we describe below.

Let $a=a(\gamma)$ be the endpoint of a curve $\gamma \in
\Gamma_1$. Let $\gamma^+$ be the directed curve that constitutes the
``right side'' of $\gamma$ oriented from $a$ to infinity and let
$\gamma^-$ be the ``left'' side of $\gamma$ oriented from infinity to
$a$. Let $a=a(\gamma)$ and $b=b(\gamma)$ be the two endpoints of a
curve $\gamma \in \Gamma_2$ that are chosen arbitrarily and fixed. Let
$\gamma^+$ (the ``right'' side) be the directed curve $\gamma$
oriented from $a$ to $b$ and let $\gamma^-$ (the ``left'' side) be
the directed curve $\gamma$ oriented from $b$ to $a$.

\subsection{Associate a sequence with a face}

Let $F_i$ of $A(\Gamma_i)$ be an unbounded face and let $C_i$ be a
connected component of $\bd F_i$. In this subsection, we show
how to associate a sequence of curves with $C_i$.
\begin{itemize}
\item For $A(\Gamma_2)$, we traverse $C_2$, keeping $F_2$ on the
  right.  Let $S_2 = \langle s_1,s_2,\ldots,s_t \rangle$ be the
  circular sequence of oriented curves in $\Gamma_2$ in the order in
  which they appear along $C_2$. If during the traversal we meet the
  curve $\gamma$ with endpoints $a=a(\gamma)$ and $b=b(\gamma)$, and
  follow it from $a$ to $b$ (respectively $b$ to $a$), we add
  $\gamma^+$ (respectively $\gamma^-$) to $S_2$.
\item Observe that for $A(\Gamma_0)$, $C_0$ is not closed---it divides
  the plane into two connected components. This means that $C_0$
  naturally corresponds to a linear sequence of un-oriented curves.
  Again, we traverse it keeping $F$ on the right.  Denote this
  sequence by $S_0$.  $S_1$ is constructed analogously, as a sequence
  of \emph{oriented} curves.
\end{itemize}

We will often abuse the notation slightly. Given a sequence of
curves $S_i$, we will often isolate an alternating subsequence, say
$A = \langle \xi_j, \gamma_j, \xi_{j+1}, \gamma_{j+1}, \ldots
\rangle$, where all $\xi_j$ represent the appearances of the same
curve $\xi$ and all $\gamma_j$ represent the appearances of the same
curve $\gamma$. We will often treat $\xi_j$ and $\gamma_j$ as 
aliases for the edges of $A(\Gamma_i)$ that correspond to those entries
in $S_i$.

\subsection{Preliminary modification for curves in $\Gamma_1$}
\label{sec:prelim2} 
Let $\Sigma_L$ and $\Sigma_R$ be the alphabets consisting of the
left symbols and right symbols, respectively. For notational
purposes, let $S_1^L = ({S_1}_{|\Sigma_L})^\Diamond$ and $S_1^R =
({S_1}_{|\Sigma_R})^\Diamond$. In this
section, we prove a key lemma, which is a variation of the
Circular Consistency Lemma \cite{DS-sharir,DS-book} below, and state an important observation.

\begin{lemma}[Linear Consistency Lemma]\hfill
  \begin{enumerate}[(a)]
  \item The portions of each arc $\xi_i^+$ appear in $S_1^R$ in
    the same order as their order along $\xi_i^+$; analogous statement holds for $S_1^L$.
  \item The portions of each arc $\xi_i$ appear in $S_0$ in the same or reverse order as compared to their
    order along $\xi_i$.
  \end{enumerate}
\end{lemma}

\begin{proof}
  We only argue (a), since (b) follows by an almost identical
  argument.
  Let $a$ and $b$ be
  portions of $\xi^+$ that occur in $S_1^R$ in that order. Assume
  that along $\xi^+$, $a$ follows $b$; refer to
  Fig~\ref{fig:ConsistencyLemma}. Denote by
  $\pi$ the portion of $C_1$ connecting $a$ to $b$. Denote
  by $\zeta$ the portion of $\xi^+$ from $b$ to $a$. Now $a\pi b\zeta$
  is a closed contour. It is easy to
  verify that the infinite ``end'' of $\xi^+$ must be enclosed in this
  closed contour, which is a contradiction.
\end{proof}

\begin{figure}
  \centering
  \includegraphics[scale=0.45]{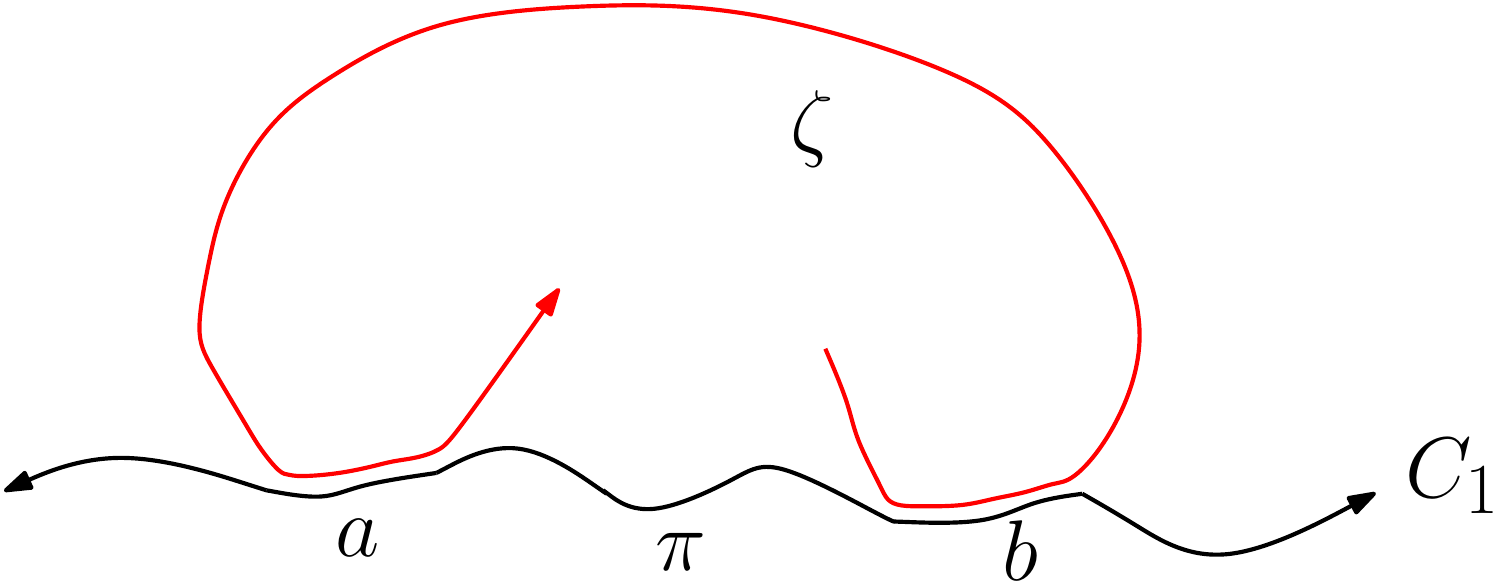}
  \caption{Proof of Linear Consistency Lemma.}
  \label{fig:ConsistencyLemma}
\end{figure}

We now make the following simple but important observation.
It can easily be checked that $S_1$ is $k$-friendly under $(\Sigma_L, \Sigma_R)$, for some $k=O(s)$.
(Indeed, the existence of a contiguous subsequence of the form
$\langle \zeta^-,\xi^+,\zeta^-,\xi^+,\ldots \rangle$ of length $s+2$, where $\zeta^- \in
\Sigma_L$ and $\xi^+ \in \Sigma_R$, would force $s+1$ distinct points
of intersection of $\zeta$ and $\xi$---a contradiction). This observation will be critical 
for the proof of Theorem~\ref{thr:unbounded}.

\subsection{Preliminary modification for curves in $\Gamma_2$} 

We will also need the following lemma.
\begin{lemma}[Circular Consistency Lemma \cite{DS-sharir,DS-book}]
  \label{lem:circular-consistency} 
  The portions of each arc $\xi_i^+$ (respectively $\xi_i^-$) appear
  in $S_2$ in a circular order consistent with their order
  along the oriented $\xi_i^+$ (respectively $\xi_i^-$). That is, there exists a starting point in $S$, 
  which depends on $\xi_i$, such that if we read $S$ in a circular order starting from that point, 
  we encounter these portions in their order along $\xi_i$. 
\end{lemma}
We now perform a cutting of the circular sequence $S_2$ as in
\cite{DS-sharir,DS-book}. Consider $S_2 = \langle s_1,\ldots,s_t
\rangle$ as a linear, rather than a circular sequence by breaking it
at an arbitrary vertex. For each directed arc $\gamma_i$, consider
the linear sequence $V_i$ of all appearances of $\gamma_i$ in $S_2$,
arranged in the order they appear along $\gamma_i$. Let $\mu_i$ and $\nu_i$ denote,
respectively, the index in $S_2$ of the first and of the last
element of $V_i$. For each arc $\gamma_i$, if $\mu_i > \nu_i$, we
split the symbol $\gamma_i$ into two distinct symbols $\gamma_{i1}$
and $\gamma_{i2}$, and replace all appearances of $\gamma_i$ in
$S_2$ between $\mu_i$ and $t$ (respectively, between 1 and $\nu_i$)
by $\gamma_{i1}$ (respectively, by $\gamma_{i2}$). Notice that by
Lemma~\ref{lem:circular-consistency},  
we are able to split $\gamma_i$ into two subarcs such that $\gamma_{i1}$ represents the 
appearances of the first subarc and $\gamma_{i2}$ represents the appearances of the second subarc. 
This splitting produces a sequence, of the same length as $S_2$ on the alphabet of
at most $4n$ symbols. To simplify the notation, hereafter we refer
to this new linear sequence as $S_2$.

\medskip

To summarize:
\begin{itemize}
\item We did not modify $S_0$. It is a linear sequence of curves.
\item $S_1$ is a linear sequence of oriented curves, from which we
  have derived two subsequences, $S_1^L$ and $S_1^R$.
\item After the cutting procedure, $S_2$ is a
  linear sequence of oriented curves.
\end{itemize}

To arrive at our first geometric theorem, we need the following
lemma \cite{DS-sharir,DS-book}.

\begin{lemma}[Quadruple Lemma \cite{DS-sharir,DS-book}] \label{lem:quadruple}
  Consider a quadruple of consecutive elements in a fixed alternating
  subsequence of $S_2$. Let this quadruple be $\langle
  \xi_1,\gamma_1,\xi_2,\gamma_2 \rangle$, such that $\xi_i$ and $\gamma_j$
  constitute portions of curves $a$ and $b$, respectively. Let $\pi_a$ be the
  portion of $a$ connecting $\xi_1$ to $\xi_2$ and let $\pi_b$ be the portion of $b$ connecting $\gamma_1$ to $\gamma_2$.
  Then $\pi_a$ and $\pi_b$ must intersect. Furthermore, this point of intersection is distinct
  for each such quadruple in this subsequence.
\end{lemma}

Although stated for $S_2$, the Quadruple Lemma also holds for $S_0$,
$S_1^R$, and $S_1^L$.

\section{Complexity of an Unbounded Face}
\label{section:unbounded}
\begin{theorem} \label{thr:unbounded} The complexity of an unbounded
  face~$F$ in an arrangement of~$n$ (0)~bi-infinite,
  (1)~semi-infinite, or (2)~bounded Jordan curves, no pair of which
  crosses more than $s$~times, is $O(\lambda_{s}(n))$,
  $O(\lambda_{s+1}(n))$, $O(\lambda_{s+2}(n))$, respectively.
\end{theorem}
\begin{proof}
  Since $\lambda_s(n)$ is at least linear, and no curve can appear on
  several connected components of $\bd F$, we can consider each
  component $C$ separately and assume that all of the curves appear on
  it. Now consider the sequences $S_0$ (case 0), $S_1^R$ (case 1), and
  $S_2$ (case 2). We claim that these are $DS(n,s)$-, $DS(2n,s+1)$-,
  and $DS(4n,s+2)$-sequences, respectively.

  Let $S$ be the sequence in question.  We aim to argue that it is a $DS$-sequence. By construction, $S$ does not contain
  any consecutive identical elements. Assume that it has an
  alternating subsequence of length $l$. Let this subsequence be $A =
  \langle \xi_1,\gamma_1,\xi_2,\gamma_2,\ldots \rangle$,
  such that $\xi_i$ and $\gamma_j$ constitute portions of curves $a$ and
  $b$, respectively.  We argue that $l$ cannot be too large.

  By Lemma~\ref{lem:quadruple}, consecutive quadruples of $A$ force
  $l-3$ distinct crossings between $a$ and $b$. 
  \begin{itemize}
  \item In case 2, setting $l = s+4$ forces $l-3 = s+1$
  distinct intersections between $a$ and $b$. This is a
    contradiction. Thus no such subsequence exists, and $S_2$ is a
    $DS(4n,s+2)$-sequence, implying that the
    complexity of $F$ is $O(\lambda_{s+2}(n))$.
  \item In case 1, setting $l = s+3$ forces $l-3 = s$ distinct
    intersections between $a$ and $b$.  Now, consider the last
    quadruple in $A$. Without loss of generality, let it be $\langle
    \xi_i, \gamma_i, \xi_{i+1}, \gamma_{i+1} \rangle$. Let $\pi$ be
    the portion of $b$ connecting $\gamma_i$ to $\gamma_{i+1}$. We
    claim that there must be an additional intersection between $a$
    and $b$ at a point on $\pi$ that we have not accounted for, so
    far.  Let $\theta \supset \xi_{i+1}$ be the portion of $C$
    connecting $\gamma_i$ to $\gamma_{i+1}$; $\pi \cup \theta$ is a
    closed contour.  Refer to Figure~\ref{fig:Semi-Infinite}. Now
    traverse $a$ in the infinite direction starting from
    $\xi_{i+1}$. Since $a$ cannot cross $\theta$, during the
    traversal, $a$ must intersect $b$ at a point on $\pi$. By the
    Linear Consistency Lemma, however this intersection has not been
    accounted for and thus there are at least $s+1$ distinct
    intersections between $a$ and $b$. This, of course, is a
    contradiction and therefore $S_1^R$ is a $DS(2n,s+1)$-sequence. By
    Theorem~\ref{fact:combination} and the discussion of
    Section~\ref{sec:prelim2}, $|S_2| = O(\lambda_{s+1}(n))$, implying
    that the complexity of $F$ is $O(\lambda_{s+1}(n))$.

    \begin{figure}
      \centering
      \includegraphics[scale=0.45]{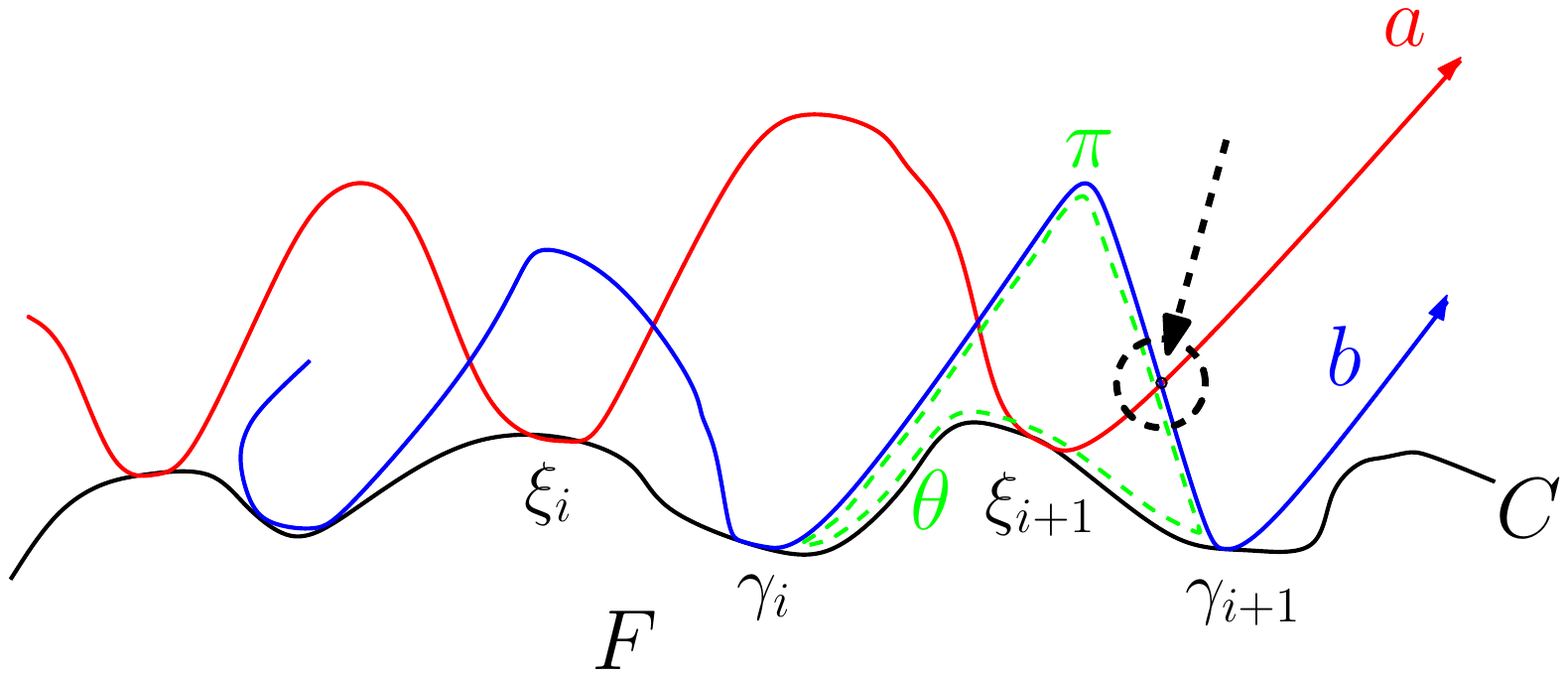}
      \caption{Semi-infinite curves.}
      \label{fig:Semi-Infinite}
    \end{figure}
  \item In case 0, setting $l = s+2$ forces $(s+2)-3 = s-1$ distinct
    intersections between $a$ and $b$.  Now, consider the first
    and last quadruples of $A$. By the same argument as
    in case 1, there are now two additional distinct intersections between
    $a$ and $b$, for a total of at least $s+1$ distinct
    intersections---a contradiction. Thus
    $S_0$ is a $DS(n,s)$-sequence and the combinatorial complexity of
    $F$ is $\lambda_{s}(n)$. Refer to Figure~\ref{fig:Biinfinite}.
    \qedhere 

    \begin{figure}
      \centering
      \includegraphics[scale=0.45]{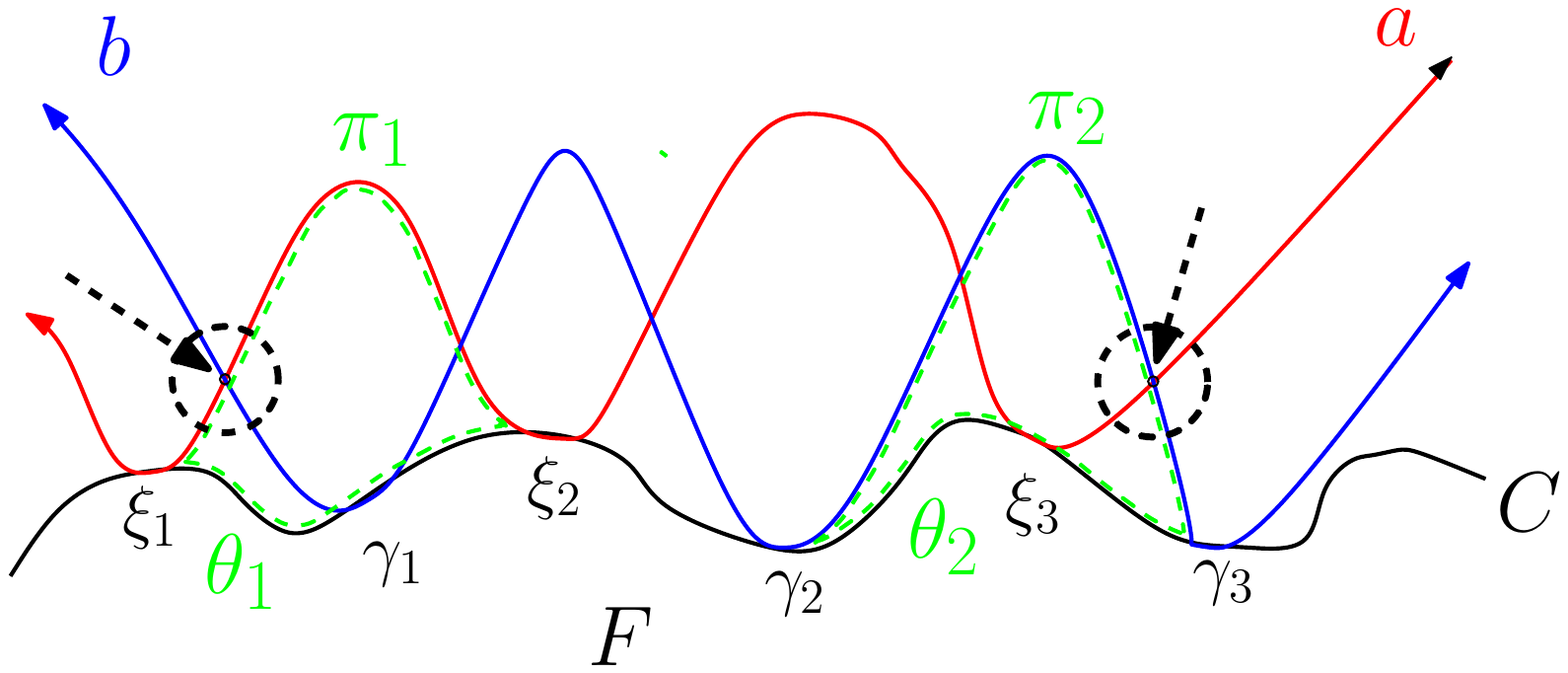}
      \caption{Bi-infinite curves.}
      \label{fig:Biinfinite}
    \end{figure} 
  \end{itemize}
\end{proof}

\section{Complexity of an Arbitrary Face} \label{section:arbitrary}

The next step is to generalize our results to bounded faces. We will
transform the problem while only increasing the complexity of a face
$F$, to make $F$ unbounded, and then apply our results from the
previous section. More precisely, we will build a tunnel from $F$ 
to the ``outside'', so that after the transformation, the new face 
will be part of an unbounded face of the arrangement.



\begin{proof}[Proof of Theorem~\ref{theorem:main}.]
  If $F$ is unbounded, we are done by Theorem~\ref{thr:unbounded}.
  Consider a bounded face $F$ in an arrangement of (a) bi-infinite,
  (b) semi-infinite, or (c) bounded Jordan curves. We assume that all 
  the curves appear on $\bd F$, otherwise one or more
  curves can be deleted without affecting the complexity of $F$.
  Furthermore since $\lambda_s(n)$ is at least linear, it is sufficient to argue the complexity of just
  one connected component of $\bd F$. Thus without loss of generality, we can assume that $\bd F$ is connected.
  \begin{description}
  \item[\emph{Step 1: Finding the site for the tunnel}] \mbox{}\\[-\parsep]
    \nopagebreak
  \begin{itemize}
  \item In cases (a) and (b), pick an arbitrary infinite edge of the
	arrangement, say of curve $a$ and follow it until it first meets
	$\bd F$, say at point $p$, where it meets curve $b$.  Denote this portion of $a$, from
	infinity to $p$, by $\zeta$; $\zeta$ is the future site for our tunnel. Refer to Figure~\ref{fig:InitialCut}(left).
    \begin{figure}
      \centering
      \includegraphics[scale=0.45]{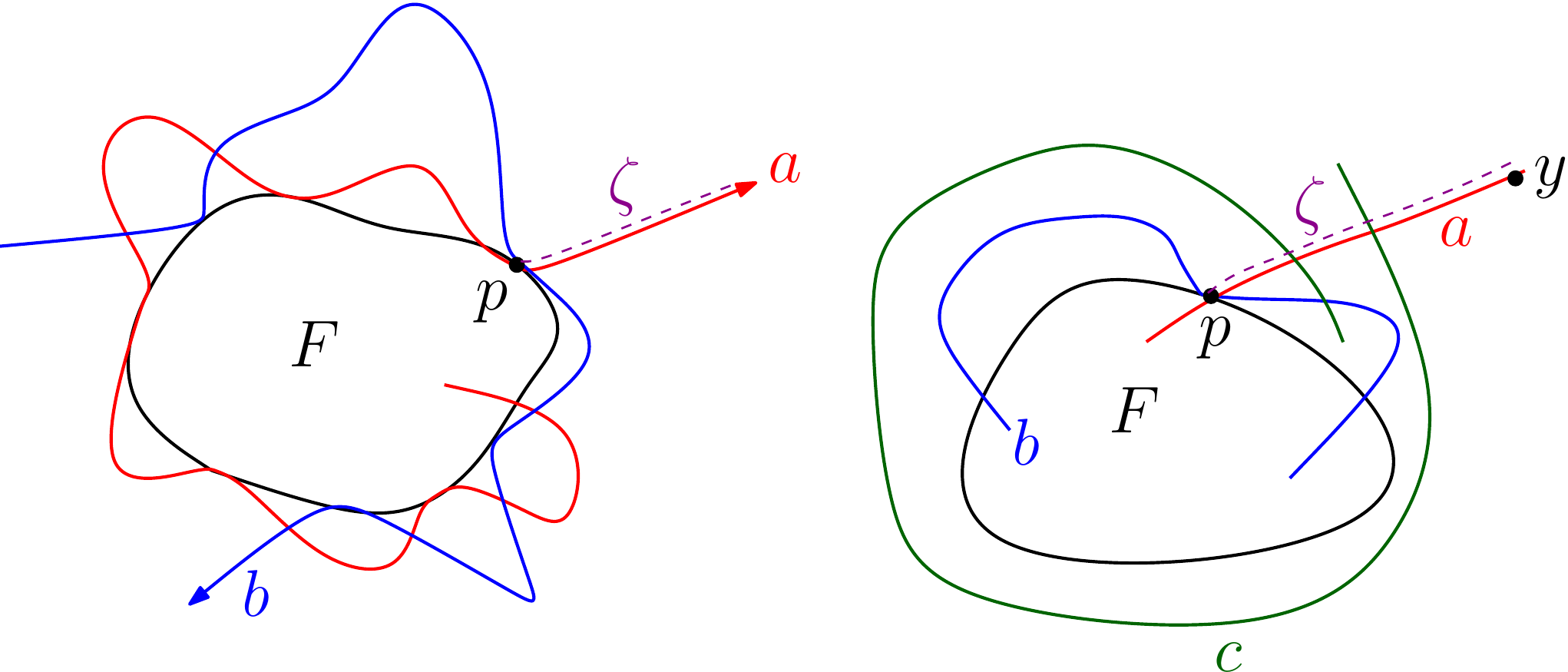}
      \caption{Finding the initial cut. $\Gamma_1$ on the left and $\Gamma_2$ on the right.}
      \label{fig:InitialCut}
    \end{figure}
  \item In case (c), if an endpoint of a Jordan arc lies on the boundary of
	the unbounded cell, we start at this point.  Otherwise, we pick an
	arbitrary edge of the infinite cell of the arrangement, cut the curve
	containing this edge into two curves, and move them slightly
        apart; this increases the number of curves by one.
	In both cases, we now have an endpoint $y$ of a curve $a$.
	Now follow $a$ from $y$ to its first point of intersection $p$ with $\bd F$,
	where it meets curve $b$. Denote this portion of $a$, from $y$
        to $p$, by $\zeta$; $\zeta$ is the future site for our
        tunnel. Refer to Figure~\ref{fig:InitialCut}(right). 
  \end{itemize}
  \item[\emph{Step 2: Digging the tunnel}] 
	We now ``dig a tunnel'' along $\zeta$ from $p$ to its ``infinite end''.
	Namely, at each intersection of $a$ with another curve $c$ of $\Gamma_i$,
	we split $c$ into two new curves, and leave a small gap between the two
	resulting curves, for $a$ to pass through (see Figure~\ref{fig:Extension}).
        \begin{figure}[h]
          \centering
          \includegraphics[scale=0.45]{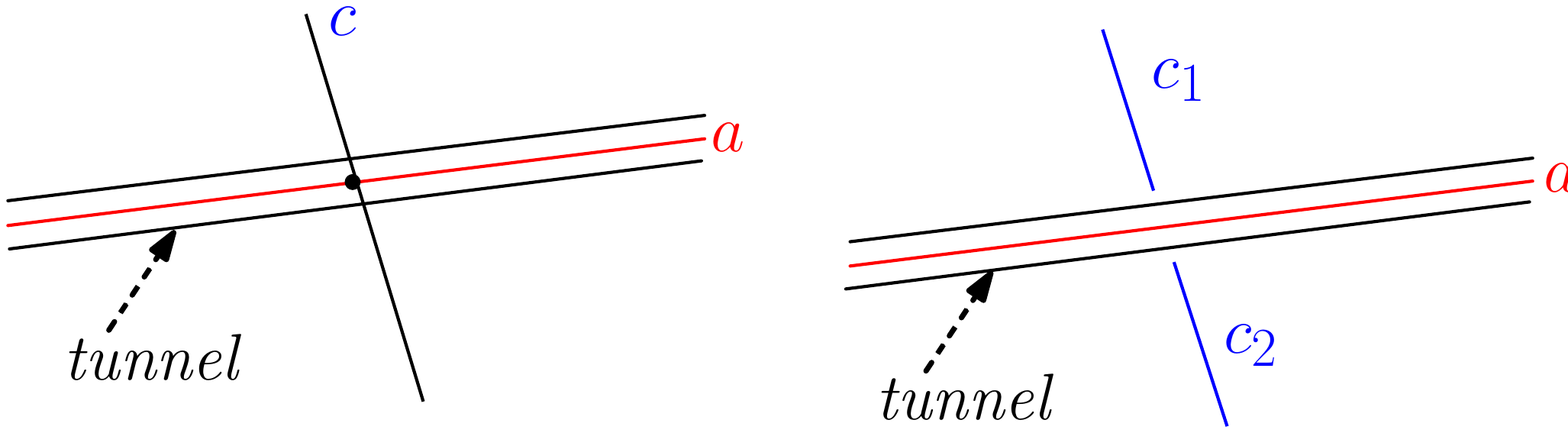}
          \caption{Extending the tunnel.  The picture before (left) and after (right).}
          \label{fig:Extension}
        \end{figure}
  \end{description}
  By construction, during our traversal of $\zeta$, $a$ did not meet
  $F$ again. Thus, as a result of our transformation we have only
  enlarged $F$, increased its complexity, and connected it to an
  infinite face. Notice that no new intersections are created. Namely,
  the resulting curves do not self-intersect and if the curves in the
  original problem intersected pairwise no more than $s$ times, then
  none of the newly created curves will intersect pairwise more than
  $s$ times. The number of curves in the resulting picture is at most
  $1 + (s+1)(n-1) = O(sn)$, if we did not have to cut at $y$; the
  remaining case is similar.
		
  We are almost done --- the trouble is that in case~(a), by splitting
  an existing curve, we cut a bi-infinite curve into semi-infinite
  curves or even finite sections; similar complications arise in
  case~(b). In case~(a), we fix this by extending infinite
  non-crossing "tails" along $a$ to infinity in such a way that they
  follow infinitesimally close to $a$ but do not cross pairwise.
  Case~(b) is handled analogously. At each split $p$, we would add one
  infinite ``tail'' to the finite sections, which were created as a
  result of the original split. Refer to Figure~\ref{fig:NewFinalCut}.
  \begin{figure}
    \centering
    \includegraphics[scale=0.45]{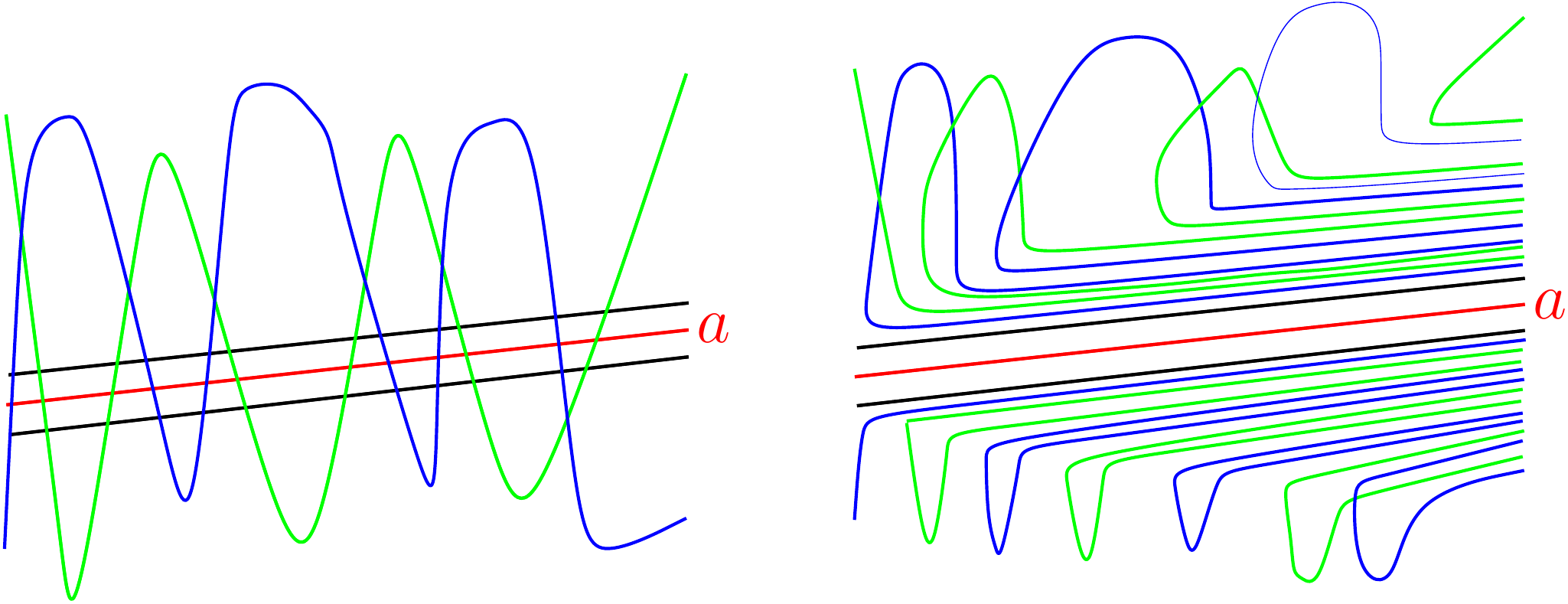}
    \caption{Fixing the tunnel. The picture before (left) and after
      (right) in the case of $\Gamma_0$, magnified.}
    \label{fig:NewFinalCut}
  \end{figure}
  We note that $\lambda_s(kn) = O(\lambda_s(n))$ for any constants $s$
  and $k$. Thus by Theorem~\ref{thr:unbounded}, the complexity of $F$
  in cases (a), (b), and (c) is $O(\lambda_s(n))$,
  $O(\lambda_{s+1}(n))$, and $O(\lambda_{s+2}(n))$, respectively.
\end{proof}

\bibliographystyle{alpha}
\small
\parsep 0pt
\bibliography{current}

\end{document}